\documentclass[conference]{IEEEtran}

\usepackage{amsthm, amssymb}

\def\snr{\textrm{SNR}}

\def\dint{{\rm \ d}}

\newtheoremstyle{slplain}
  {2pt}
  {2pt}
  {\slshape}
  {}
  {\bfseries}
  {.}%
  { }
  {}

\theoremstyle{slplain}

\newtheorem{cor}{Corollary}
\newtheorem{lem}{Lemma}
\newtheorem{pro}{Proposition}

\usepackage{cite}
\usepackage{amsfonts}
\usepackage{multicol,multienum}
\usepackage{subfigure}
\usepackage{multirow}

\ifCLASSINFOpdf
  \usepackage[pdftex]{graphicx}
  \graphicspath{{../pdf/}{../jpeg/}}
  \DeclareGraphicsExtensions{.pdf,.jpeg,.png}
\else
  \usepackage{float}
  \usepackage[dvips]{graphicx}
  \graphicspath{{../eps/}}
  \DeclareGraphicsExtensions{.eps}
\fi

\usepackage[cmex10]{amsmath}
\usepackage{algorithm}
\usepackage{algorithmic}
\usepackage{array}
\usepackage{url}
\usepackage{setspace}

\begin{document}

\title{Throughput Optimal Listen-Before-Talk for Cellular in Unlicensed Spectrum}
%\title{Optimizing Listen-Before-Talk to Maximize Transmission Throughput for Load Based Equipment in Unlicensed Spectrum}

%\author{
%%\IEEEauthorblockA{Ning Wei and Xingqin Lin}
%\IEEEauthorblockA{Xingqin Lin}
%\thanks{Xingqin Lin is with Ericsson Research, San Jose, CA, USA. (Email: xingqin.lin@ericsson.com.)}
%%}
%%\thanks{N. Wei and Z. Zhang are with the National Key Laboratory of Science and Technology on Communications,
%%University of Electronic Science and Technology of China, Sichuan, China.  (Email: \{wn, zhangzp\}@uestc.edu.cn.)  X. Lin is with Ericsson Research, San Jose, CA, USA. (Email: xingqin.lin@ericsson.com.)
%%}
%}

%\author{\IEEEauthorblockN{Ning Wei}
%\IEEEauthorblockA{National Key Laboratory of Science and Technology on Communications\\
%University of Electronic Science and Technology of China\\
%Chengdu, Sichuan, China\\
%Email: wn@uestc.edu.cn}
%\and
%\IEEEauthorblockN{Xingqin Lin}
%\IEEEauthorblockA{Ericsson Research\\
%Santa Clara, California, USA\\
%Email: xingqin.lin@ericsson.com}
%}

\author{\IEEEauthorblockN{Ning Wei\IEEEauthorrefmark{1}, Xingqin Lin\IEEEauthorrefmark{2}, Wanwan Li\IEEEauthorrefmark{1},Youzhi Xiong\IEEEauthorrefmark{1}, Zhongpei Zhang\IEEEauthorrefmark{1}}
\IEEEauthorblockA{\IEEEauthorrefmark{1}National Key Laboratory of Science and Technology on Communications\\
University of Electronic Science and Technology of China,
Chengdu, Sichuan, China}
%Email: \{wn, zhangzp\}@uestc.edu.cn}
\IEEEauthorblockA{\IEEEauthorrefmark{2}Ericsson Research,
Santa Clara, California, USA\\
Contact email: wn@uestc.edu.cn}
%Email: xingqin.lin@ericsson.com}
%\thanks{This work was supported by National Natural Science Foundation of China (Grant Nos. 61101092, 61571003), National Science and Technology Major Project of China (Grant No. 2012ZX03001003-003), National High-Tech Research Development Program of China (863) (Grant Nos. 2014AA01A704, 2014AA01A706), and the Fundamental Research Funds for the Central Universities (Grant No. ZYGX2015J014).}
}

\maketitle

\begin{abstract}
The effort to extend cellular technologies to unlicensed spectrum has been gaining high momentum. Listen-before-talk (LBT) is enforced in the regions such as European Union and Japan to harmonize coexistence of cellular and incumbent systems in unlicensed spectrum. In this paper, we study throughput optimal LBT transmission strategy for load based equipment (LBE). We find that the optimal rule is a pure threshold policy: The LBE should stop listening and transmit once the channel quality exceeds an optimized threshold. We also reveal the optimal set of LBT parameters that are compliant with regulatory requirements. Our results shed light on how the regulatory LBT requirements can affect the transmission strategies of radio equipment in unlicensed spectrum.
\end{abstract}


%\IEEEpeerreviewmaketitle

\section{Introduction}

There has been a surge of interest in extending cellular technologies to unlicensed spectrum that has been heavily used by wireless local area network (WLAN) technologies \cite{ieee802dot11}. The initial interest from cellular industry in unlicensed spectrum started when researchers realized that WLAN is a powerful offloading technique that can greatly help alleviate network congestion due to the exponential growth of mobile traffic demand \cite{dimatteo2011cellular, lee2013mobile}. The more recent interest has been shifted to directly adapting cellular technologies, particularly Long Term Evolution (LTE) developed under the 3rd Generation Partnership Project (3GPP), to unlicensed spectrum. Exemplary efforts include LTE-WLAN Aggregation (LWA) \cite{lwa2015wi}, Licensed-Assisted Access (LAA) \cite{laa2015wi}, LTE Unlicensed (LTE-U) \cite{lteuwebsite}, and MuLTEfire \cite{multefirewebsite}. Making use of the readily available unlicensed spectrum is a natural choice to boost cellular network capacity.

Despite the great potential of utilizing unlicensed spectrum for mobile broadband communications, many concerns have been raised. One critical concern is about how to ensure harmonious coexistence and fair sharing of unlicensed spectrum between cellular and other incumbent systems in unlicensed spectrum \cite{zhang2015lte}. To facilitate coexistence, a radio equipment operating in  unlicensed spectrum shall comply with regulatory requirements, which vary by regions. Typical regulatory requirements include maximum in-band output power, in-band power spectral density, out-of-band spurious emissions, radar protection, and channel access methods \cite{regulation3gpp}. Among the regulatory requirements, channel access methods such as the requirement of LBT are perhaps the most controversial. While not required in the regions such as United States, Korea, and India, LBT is enforced in the regions such as European Union and Japan \cite{regulation3gpp}. As a result, a truly global cellular technology in unlicensed spectrum shall incorporate LBT features.

Many research works have been triggered to study the coexistence of LTE and WLAN in unlicensed spectrum. A few studies have examined the impact of LBT strategies on the coexistence of cellular and WLAN in unlicensed spectrum \cite{zhang2015hierarchical, chen2015downlink, chen2016energy, lien2016random, ko2015fair, li2015modeling}. These studies confirm that appropriate LBT schemes help ensure fair sharing of unlicensed spectrum. Nevertheless, deeper understanding of LBT is still required and in particular it would be highly desirable to carefully examine existing LBT regulatory requirements. Such studies could potentially help regulators evolve spectrum policies and also could be useful for network operators or equipment vendors to optimize radio transmissions in unlicensed spectrum.

In this paper, we focus on LBT requirements interpreted by the European Telecommunications Standards Institute (ETSI) \cite{etsi2015en} and apply optimal stopping theory to study throughput optimal transmission strategy for load based equipment (LBE). Optimal stopping theory is about determining a time to take a given action based on causal observations to maximize an expected reward \cite{ferguson2012optimal}. Optimal stopping theory has recently been applied to study emerging problems in wireless communications and networking such as wireless power transfer \cite{yang2014dynamic},  millimeter-wave cellular systems \cite{wei2015optimal}, and energy harvesting based wireless networks \cite{li2015distributed}. Using optimal stopping theory, we show that the throughput optimal strategy is a pure threshold policy: The LBE should stop listening and transmit once the channel quality exceeds an optimized threshold. We also reveal the optimal set of LBT parameters that are compliant with the regulatory requirements.

%The rest of this paper is organized as follows. Section \ref{sec:model} describes the mathematical model and formulates the problem. Throughput optimal LBT for LBE is analyzed and derived in Section \ref{sec:analysis}. Section \ref{sec:sim} presents demonstrating numerical and simulation results, and is followed by our concluding remarks in Section \ref{sec:conclusions}.

\section{Mathematical Model and Problem Formulation}
\label{sec:model}

\subsection{Channel Access Mechanism of Load Based Equipment}

%According to ETSI \cite{etsi2015en}, channel access mechanism of LBE (option B) is based on energy detection. Before a transmission or burst of transmissions on an \textit{operating channel}, a LBE shall perform a \textit{clear channel assessment} (CCA) to check whether the energy level at the input of the receiver exceeds the threshold level
%\begin{align}
%\textrm{TL}= - 73 + \max(23-P_{\textrm{H}}, 0)  \quad  \textrm{(dBm/MHz)}
%\end{align}
%where $P_{\textrm{H}}$ is the maximum transmit power in dBm and a receive antenna with $0$ dBi antenna gain is assumed. The CCA observation time shall be at least $20$ us \cite{etsi2015en}.

According to ETSI \cite{etsi2015en}, channel access mechanism of LBE (option B) is based on energy detection. Before a transmission or burst of transmissions on an \textit{operating channel}, a LBE shall perform a \textit{clear channel assessment} (CCA) to check whether the energy level at the input of the receiver exceeds some threshold level. The CCA observation time shall be at least $20$ $\mu s$ \cite{etsi2015en}. If the receive energy level does not exceed the threshold level, the operating channel is considered clear and the LBE may transmit immediately on the operating channel for a \textit{channel occupancy time} without re-performing CCA in the mean time. Otherwise, the operating channel is considered occupied and the LBE shall perform extended CCA (ECCA) checks. The ECCA check procedure goes as follows \cite{etsi2015en}.
\begin{enumerate}
\item The LBE randomly generates an integer in the range $1$ to $q$ and stores the value in a counter denoted by $Z$. The value $q$ can be selected by the equipment manufacturer from the range $4$ to $32$.
\item For every ECCA observation time, denoted by $T_{\textrm{ecca}}$, if the operating channel is detected to be clear, decrement the counter $Z$ by one. If the operating channel is detected to be occupied, the counter $Z$ remains unchanged.
\item Once the counter $Z$ reaches zero, the LBE may initiate transmission for a channel occupancy time without re-performing CCA or ECCA in the mean time.
\end{enumerate}

The maximum channel occupancy time, denoted by $T_{\textrm{cot}, \max}$, shall be less than $\frac{13}{32} \times q$ $m s$ \cite{etsi2015en}. After the maximum channel occupancy time, if the LBE needs another transmission, it shall go through ECCA check procedure again. It follows that the ECCA phase is always required except for the case where the initial CCA check is clear and the LBE chooses to transmit immediately. In this paper, we assume that each burst consists of many transmissions and thus the heterogeneous behavior of the first transmission may be ignored. In other words, we assume that each transmission starts directly with an ECCA phase.
%Note that each ECCA phase may consist of one or more ECCA checks that depend on the realization of the counter and radio environment.

\subsection{Transmission Strategies}

A LBE may transmit after it completes an ECCA phase, but it does not have to. In particular, the radio channel may be currently in deep fade, and thus it might be better off for the LBE to skip the current transmission opportunity. Motivated by this observation, we propose the following regulation compliant transmission protocol. First, the LBE performs ECCA per the regulations. Once an ECCA phase is finished, the LBE performs a further tentative probing of its link quality. Depending on the outcome of the probing, the LBE decides whether or not to proceed with data transmission.

For the link probing, the transmitter may transmit a short known preamble, based on which the receiver can measure the the link quality and feeds back the result to the transmitter. We assume that each probing process requires a time duration $\tau T_{\textrm{cot},\max}$, where $\tau \in [0, 1]$, leaving the remaining time $(1-\tau) T_{\textrm{cot}, \max}$ in the allowed channel occupancy duration for potential data transmission. Therefore, the communication process is periodic in time and each period consists of two steps: channel assessment (including ECCA and link probing) and data transmission. Accordingly, the channel occupancy time $T_{\textrm{cot}}$ is given by
\begin{align}
T_{\textrm{cot}} = \left\{ \begin{array}{ll}
         \tau T_{\textrm{cot}, \max} + (1-\tau) T_{\textrm{cot}, \max} & \mbox{if to transmit};\\
         \tau T_{\textrm{cot}, \max} & \mbox{otherwise}.\end{array} \right.
\end{align}

The outcome of an ECCA check highly depends on the radio environment. In the unlicensed spectrum, there may exist transmissions from miscellaneous sources whose activities are scenario dependent and highly unpredictable. In this paper, for analytical tractability, we assume that the outcomes of ECCA checks are independent and identically distributed (i.i.d.) across observations, and that the probability of observing a clear ECCA is $p \in (0, 1]$.

Figure \ref{fig:1} shows a sample realization of channel assessment and data transmission in a communication period. In the first ECCA phase, the counter $Z=1$ and is decremented to $0$ with a clear ECCA check. The LBE then probes the link but decides to skip the transmission opportunity. The LBE enters the second ECCA phase with $Z=2$ afterwards. After obtaining $2$ clear ECCA checks out of $4$ ECCA observation periods, the counter is decremented to $0$. The LBE probes the link again and then decides to proceed with data transmission.

\begin{figure}
\centering
\includegraphics[width=8.8cm]{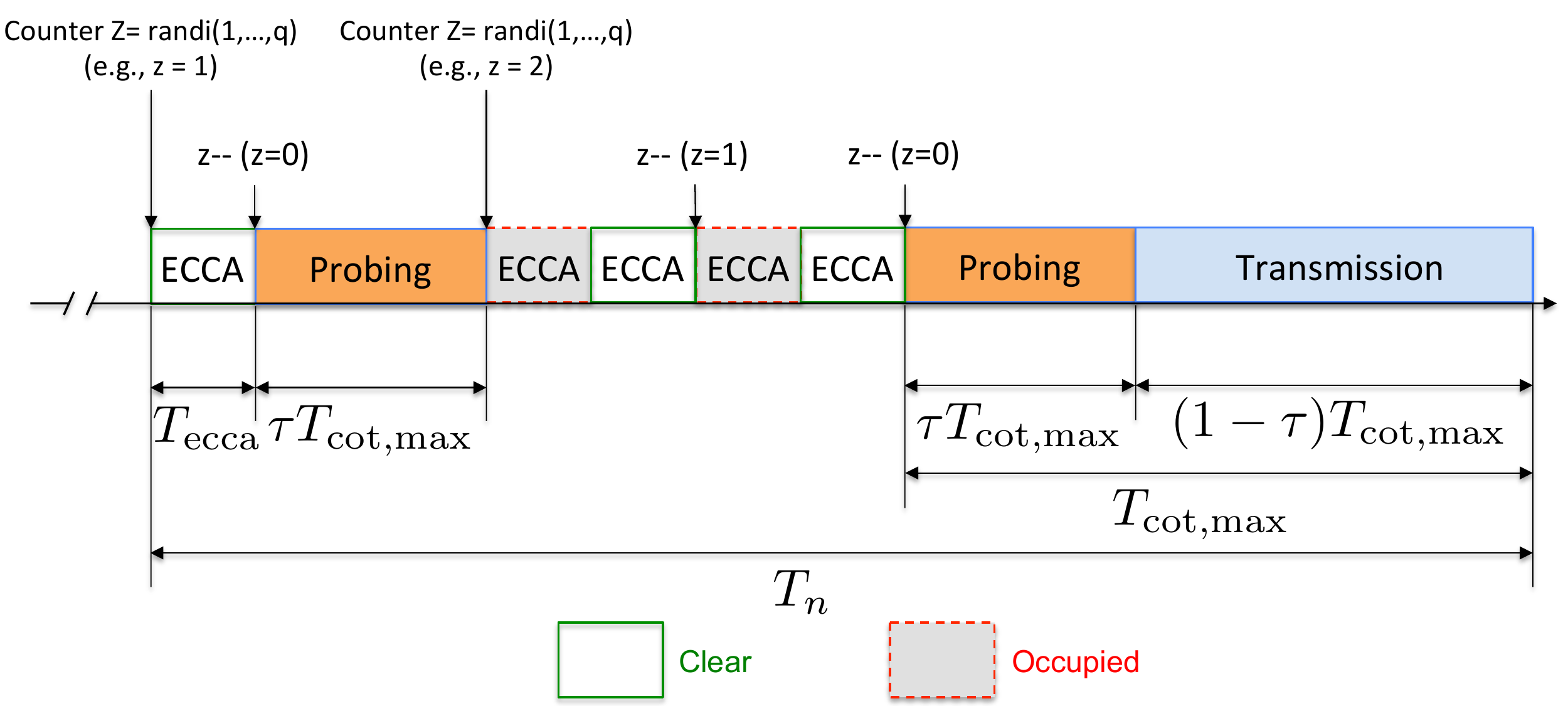}
\caption{A sample realization of channel assessment and data transmission in a communication period.}
\label{fig:1}
\end{figure}

\subsection{Problem Formulation}

Consider a typical communication period with $n$ ECCA phases, and let $T_n$ be the duration of the period.  Denote by $R_n$ the non-negative spectral efficiency (bps/Hz) of the LBE link after $n$ ECCA phases.  We assume that $R_1, R_2, ...$ are i.i.d. and denote by $R$ the generic random variable for $\{R_i\}$ with cumulative distribution function $F_R(x)$.
The spectral efficiency is a function of signal-to-noise ratio (SNR), and may be given by
$
R_n = \log ( 1 + \snr_n ).
$
The analysis in this paper however only requires that the second moment of $R_n$ exists.

If the LBE decides to proceed with data transmission after $n$ ECCA phases, the amount of bits that can be delivered equals
\begin{align}
U_n = (1-\tau) T_{\textrm{cot}, \max}  W R_{n} ,
\end{align}
where $W$ denotes the bandwidth of the operating channel.
%Since $R_1, R_2, ...$ are i.i.d., $U_1, U_2, ...$ are also i.i.d.. Similarly, denote by $U$ the generic random variable for $\{U_i\}$  with cumulative distribution function $F_U(x)$.
If we repeatedly use a stopping rule across $m$ communication periods, the total number of bits that can be delivered equals
$\sum_{i=1}^m U_{n(i)}$, where $n(i)$ denotes the number of ECCA phases completed in the $i$-th period. Accordingly, the duration of the $m$ periods equals $\sum_{i=1}^m T_{n(i)}$. Therefore, the throughput (bit/s) equals $\sum_{i=1}^m U_{n(i)}/\sum_{i=1}^m T_{n(i)}$. Letting $m\to\infty$, by the law of large numbers, the ergodic throughput equals ${ \mathbb{E}[ U_N ] }/{\mathbb{E}[ T_N ]}$.

Denoting the set of admissible stopping rules by
\begin{align}
\mathcal{C}=\{N \in \mathbb{N}^+: \mathbb{E}[T_N] < +\infty\},
\end{align}
our objective is to find an optimal stopping rule $N^\star$ to obtain the maximum ergodic throughput $\lambda^\star$. Mathematically, the throughput optimization problem is written as follows.
\begin{align}
\lambda^\star \triangleq \sup_{N \in \mathcal{C}} \frac{ \mathbb{E}[ U_N ] }{\mathbb{E}[ T_N ]}.
\label{eq:3}
\end{align}

%Burst of $X$ transmissions. $X=1$ and $X\to \infty$ are interesting special cases. $X_i \sim X \in [1,...,M]$ transmissions
%
%$$
%U_{n(j, i)} = W (1-\tau) T_{\textrm{cot}, \max} R_{n(j, i)}
%$$
%
%
%$$
%T_{n(j, i)}
%$$
%
%Throughput:
%$$
%\mu = \mu_m \mathbb{P}(X=m)
%$$
%
%where
%$$
%\mu_m (N) = \lim_{K \to \infty} \frac{ \sum_{i=1}^{K} \sum_{j=1}^{m} U_{N(j, i)} }{ \sum_{i=1}^{K} \sum_{j=1}^{m} T_{N(j, i)} }
%$$
%
%don't know how to do this: $\max_N \mu_m (N)$. don't even know how to simulate this.
%
%A separate approach: separately find optimal stopping rules for 1st and (2nd, 3rd,...) transmissions.
%
%$$
%\max_{N^{(\textrm{cca})}} \frac{ \mathbb{E} [ U_{N^{(\textrm{cca})}}  ]  }{ \mathbb{E} [ T_{N^{(\textrm{cca})}} ] }
%$$
%
%optimal for $X\equiv 1$
%
%$$
%\max_{N^{(\textrm{ecca})}} \frac{ \mathbb{E} [ U_{N^{(\textrm{ecca})}}  ]  }{ \mathbb{E} [ T_{N^{(\textrm{ecca})}} ] }
%$$
%
%optimal for $X\to \infty$

\section{Throughput Optimal Listen-Before-Talk}
\label{sec:analysis}

%\subsection{Threshold Policy Achieves the Optimal Throughput}

In this section we characterize the throughput optimal stopping strategy $N^\star$ and the attained maximum throughput $\lambda^\star$. To this end, we first derive the expected duration of a typical communication period under the transmission protocol in Lemma \ref{lem:1}.

\begin{lem}
The expected duration $\mathbb{E}[T_n]$ of a communication period with $n$ ECCA phases is given by
\begin{align}
\mathbb{E} [ T_n ] = T_{\textrm{cot}, \max}  + (n - 1) \tau T_{\textrm{cot}, \max} +  \frac{n(q+1)}{2p} T_{\textrm{ecca}} .
\label{eq:6}
\end{align}
\label{lem:1}
\end{lem}
\begin{proof}
Denote by $X_i$ the number of ECCA checks in the $i$-th ECCA phase.  The duration of a communication period with $n$ ECCA phases can be written as
\begin{align}
T_n &= \sum_{i=1}^n \left( X_i  T_{\textrm{ecca}}    +   \tau T_{\textrm{cot}, \max} \right) +  (1-\tau) T_{\textrm{cot}, \max}  \notag \\
&=T_{\textrm{cot}, \max}  + (n - 1) \tau T_{\textrm{cot}, \max} +  T_{\textrm{ecca}}  \sum_{i=1}^n X_i .
\label{eq:4}
\end{align}
Since the outcomes of ECCA checks are i.i.d., $\{X_i\}$ are i.i.d.. Conditioning on the counter $Z=z$, the number of ECCA checks in the $i$-th ECCA phase is distributed as
\begin{align}
\mathbb{P}(X_i = z + k) = \binom{z+k-1}{z-1} p^{z-1} (1-p)^{k} p, \quad k = 0, 1, ....\notag
\end{align}
It follows that
\begin{align}
\mathbb{E} [ X_i ] = \mathbb{E} [ \mathbb{E} [ X_i|Z ]  ]  = \mathbb{E} \left[ \frac{Z}{p} \right] =  \frac{q+1}{2p}, \quad \forall i.
\label{eq:5}
\end{align}
Plugging (\ref{eq:5}) into (\ref{eq:4}) yields (\ref{eq:6}).
\end{proof}

While Lemma \ref{lem:1} characterizes the expected duration of a communication period with $n$ ECCA phases, different stopping rules yield different numbers of ECCA phases that possibly vary across communication periods. With Lemma \ref{lem:1}, we are in a position to derive the optimal stopping rule in Proposition \ref{pro:1}.

\begin{pro}
The optimal stopping rule for the throughput maximization problem (\ref{eq:3}) is given by
\begin{align}
N^\star = \min \left\{ n \geq 1: R_n  \geq  \frac{\lambda^\star}{W} \right\},
\label{eq:12}
\end{align}
where $\lambda^\star$ is the unique maximum throughput. Further, $\lambda^\star$ is the solution to the following fixed point equation:
\begin{align}
\mathbb{E} \left[ \left( WR  - \lambda    \right)^+  \right]
= \zeta \lambda  ,
\label{eq:11}
\end{align}
where $(x)^+ \triangleq \max(x, 0)$ and
\begin{align}
\zeta \triangleq \frac{ \tau T_{\textrm{cot}, \max} + \frac{q+1}{2p} T_{\textrm{ecca}}  }{ (1-\tau) T_{\textrm{cot}, \max}   } .
\label{eq:13}
\end{align}
\label{pro:1}
\end{pro}
\begin{proof}
We first solve the associated ordinary optimal stopping problem.
\begin{align}
V( \lambda ) &\triangleq  \sup_{N \in \mathcal{C}}   \mathbb{E} [U_N - \lambda T_N ]  ,
\label{eq:2}
\end{align}
where $\lambda$ is an arbitrary positive number.
The duration of the first ECCA phase equals
$
\tau T_{\textrm{cot}, \max} + X_1 T_{\textrm{ecca}} .
$
If the LBE decides to transmit, it can deliver $U_1$ information bits. If instead the LBE skips this transmission opportunity and continues to listen from this point, the $U_1$ bits are not transmitted and the time $\tau T_{\textrm{cot}, \max} + X_1 T_{\textrm{ecca}}$ has passed. In the second ECCA phase, the problem starts over again, implying that the problem is invariant in time.
%Generalizing this to stage $n$, if the LBE continues to listen, it can obtain an expected return of $V^\star$ but no more.  Therefore, $\mathbb{E}[ V_{n+1}^\star |  \mathcal{F}_n ] = \sup_{N> n} \mathbb{E}[ Y_N | \mathcal{F}_n ]= V^\star$.
Therefore, without loss of generality, we can focus on the first ECCA phase to compute $V(\lambda)$.
%This may be done through the optimality
%equation.
Specifically, if the LBE stops and transmits, it obtains a utility of $U_1 - \lambda T_1 $. If the LBE continues to listen, it obtains a utility of $V(\lambda)- \lambda ( \tau T_{\textrm{cot}, \max} + X_1 T_{\textrm{ecca}})$. By the optimality equation of dynamic programming,
\begin{align}
V (\lambda) =& \mathbb{E} \left[ \max ( U_1 - \lambda T_1, V (\lambda)  - \lambda ( \tau T_{\textrm{cot}, \max} + X_1 T_{\textrm{ecca}}) ) \right] \notag \\
=& \mathbb{E} \left[ \max ( U_1  -  \lambda (1-\tau) T_{\textrm{cot}, \max},  V (\lambda)  )  \right] \notag \\
&-   \lambda \left( \tau T_{\textrm{cot}, \max} + \frac{q+1}{2p} T_{\textrm{ecca}} \right) ,
\label{eq:7}
\end{align}
where we have used $\mathbb{E}[X_1] = \frac{q+1}{2p}$ (c.f. (\ref{eq:5})) in the last equality.
Rearranging the terms in (\ref{eq:7}) yields that
\begin{align}
\mathbb{E} \left[ ( U  - \lambda  (1-\tau) T_{\textrm{cot}, \max}  - V (\lambda)  )^+  \right] \notag \\
= \lambda \left( \tau T_{\textrm{cot}, \max} + \frac{q+1}{2p} T_{\textrm{ecca}} \right) .
\label{eq:8}
\end{align}
%The left side is continuous in $V^\star$ and decreasing from $+\infty$ to zero, while the right side is positive. Hence, there is a unique solution $V (\lambda)$.
To obtain the maximum utility $V(\lambda)$, the LBE can stop once the currently achievable utility is not less than the maximum expected utility that is obtained with continuing.
Therefore, the optimal stopping strategy is given by
\begin{align}
N^\star (\lambda)  = \min \{ n \geq 1: U_n  \geq  \lambda (1-\tau) T_{\textrm{cot}, \max} + V (\lambda) \} .
\label{eq:9}
\end{align}

By Theorem 1, Chapter 6 in \cite{ferguson2012optimal}, we know that  $N^\star$ is an optimal stopping rule that attains the maximum throughput $\lambda^\star$ in the throughput optimization problem (\ref{eq:3}) if and only if $N^\star$ is an optimal stopping rule for the ordinary optimal stopping problem (\ref{eq:2}) with $\lambda=\lambda^\star$ and $V(\lambda^\star) = 0$. Plugging $V(\lambda) = 0$ and $U = (1-\tau) T_{\textrm{cot}, \max}  W R$ into (\ref{eq:8})  yields (\ref{eq:11}).
The left side of (\ref{eq:11}) is continuous in $\lambda$ and decreasing from $\mathbb{E}[R^+]$ to zero, while the right side of (\ref{eq:11}) is continuous in $\lambda$ and increasing from $0$ to $+\infty$. Hence, there is a unique solution $\lambda^\star$. Letting $V(\lambda) = 0$ in (\ref{eq:9}) yields the optimal stopping strategy in (\ref{eq:12}). This completes the proof.
\end{proof}

Proposition \ref{pro:1} implies that the optimal stopping rule is a pure threshold policy: The LBE should take the opportunity to transmit in the operating channel once the currently achievable spectral efficiency of the operating channel exceeds $\lambda^\star/W$. The threshold $\lambda^\star/W$ increases with $\lambda^\star$. This implies that if the maximum throughput is larger, the LBE can be less aggressive in transmitting in the operating channel by using a higher stopping threshold.

Under optimal stopping, the following Corollary \ref{cor:1} readily follows from Proposition \ref{pro:1}.
\begin{cor}
The expected time of a communication period under optimal stopping is given by
\begin{align}
\mathbb{E} [ T_{N^\star} ] = (1 - \tau) T_{\textrm{cot}, \max} + \frac{\tau T_{\textrm{cot}, \max} +  \frac{ (q+1)}{2p} T_{\textrm{ecca}}}{ 1-F_R( {\lambda^\star}/{W} ) }.
\label{eq:14}
\end{align}
The expected number of bits that can be delivered in a communication period under optimal stopping equals
\begin{align}
&\mathbb{E} [ U_{N^\star} ] = (1-\tau) T_{\textrm{cot}, \max}  W \int_{\lambda^\star/W}^\infty \frac{ F_{R} (x) - F_R( {\lambda^\star}/{W} )  }{ 1-F_R( {\lambda^\star}/{W} ) } \dint x.
\label{eq:15}
\end{align}
\label{cor:1}
\end{cor}
\begin{proof}
%We omit the proof due to page constraint.
From Lemma \ref{lem:1}, $\mathbb{E} [ T_{N^\star} ] $ equals
\begin{align}
 (1 - \tau) T_{\textrm{cot}, \max} + ( \tau T_{\textrm{cot}, \max} +  \frac{ (q+1)}{2p} T_{\textrm{ecca}} ) \mathbb{E} [N^\star] .
\label{eq:17}
\end{align}
From Proposition \ref{pro:1}, the stopping time $N^\star$ (i.e., the number of ECCA phases passed before data transmission) is geometrically distributed with parameter $1-F_R( {\lambda^\star}/{W} )$, and thus
$
\mathbb{E} [N^\star]  = \frac{1}{ 1-F_R( {\lambda^\star}/{W} ) }.
$
Plugging $\mathbb{E} [N^\star] $ into (\ref{eq:17}) yields (\ref{eq:14}). Also from Proposition \ref{pro:1} the stopped random variable $R_{N^\star}$ has the following distribution:
\begin{align}
F_{R_{N^\star}} (x) = \frac{ F_{R} (x) - F_R( {\lambda^\star}/{W} )  }{ 1-F_R( {\lambda^\star}/{W} ) }, \quad x\geq {\lambda^\star}/{W},
\label{eq:18}
\end{align}
from which we can compute $\mathbb{E}[R_{N^\star}]$. Plugging $\mathbb{E}[R]$ into $\mathbb{E} [ U_{N^\star} ] =  (1-\tau) T_{\textrm{cot}, \max}  W \mathbb{E}[R_{N^\star}]$ yields (\ref{eq:15}).
\end{proof}

Though in general the maximum throughput $\lambda^\star$ does not admit a closed form,  a careful examination of the fixed point equation (\ref{eq:11}) reveals that \textit{the solution $\lambda^\star$ monotonically increases as the factor $\zeta$ decreases.} This key observation leads to the discoveries of the impact of the various system and regulation parameters on the maximum achievable throughput. Several remarks are in order.

\textbf{Remark 1.} $\lambda^\star$ increases with $T_{\textrm{cot}, \max}$ and $p$ but decreases with $T_{\textrm{ecca}}$ and $\tau$. It is intuitive that longer channel occupancy time $T_{\textrm{cot}, \max}$ or smaller link probing overhead $\tau$ yields larger throughput. If the probability $p$ of observing a clear ECCA is larger or the duration $T_{\textrm{ecca}}$ of an ECCA observation time is shorter, the duration of ECCA phases decreases, leading to increased throughput.

\textbf{Remark 2.} The dependency of $\lambda^\star$ on the maximum counter value $q$ hinges on whether or not we scale $T_{\textrm{cot}, \max}$ with $q$. If $T_{\textrm{cot}, \max}$ is fixed, choosing a smaller $q$ reduces the duration of  ECCA phases and the throughput increases. However, it is regulated that $T_{\textrm{cot}, \max}$ shall be less than $\frac{13}{32} q$ \cite{etsi2015en}. Thus, choosing a smaller $q$ also reduces the allowed channel occupancy time that may in turn reduce throughput. It is a tradeoff between ``listening'' and ``talking.''  With a given $q$, the throughput is maximized with the largest allowed channel occupancy time. Plugging the supremum $\frac{13}{32} q$ of the allowed channel occupancy times into (\ref{eq:13}) yields
\begin{align}
\zeta = \frac{ \tau  + \frac{16}{13p}T_{\textrm{ecca}} (1 + {1}/{q}) }{ 1-\tau    } .
\end{align}
It follows that choosing a larger $q$ increases the throughput. We summarize this result in Proposition \ref{pro:2}.

\begin{pro}
With the regulation constraints that $T_{\textrm{cot}, \max} < \frac{13}{32} q$ $m s$ and $4\leq q \leq 32$, the optimal throughput $\lambda^\star$ is maximized with $q=32$ and $T_{\textrm{cot}, \max} = 13 - \epsilon$ $m s$, where $\epsilon$ is an arbitrarily small positive number.
\label{pro:2}
\end{pro}

Before ending this section we devise a numerical iterative algorithm to solve the equation (\ref{eq:11}). To this end, note that by definition the maximum throughput $\lambda^\star = \frac{\mathbb{E} [ U_{N^\star} ] }{ \mathbb{E} [ T_{N^\star} ]   }$. Plugging (\ref{eq:14}) and (\ref{eq:15}) into this definition yields that
\begin{align}
\lambda^\star = g(\lambda^\star),
\label{eq:16}
\end{align}
where
\begin{align}
g(x) = \frac{ W \int_{x/W}^\infty \left( F_{R} (r) - F_R( x/W ) \right)  \dint r  }{ 1-F_R( x/W )  +  \zeta }.
\end{align}
Inspired by (\ref{eq:16}), we propose the following iterative algorithm:
\begin{align}
\lambda [t + 1] = g( \lambda [t]   ) , \quad t = 0, 1, ....
\end{align}
This iterative method is a variation of Newton's method and converges for any non-negative initial value $\lambda[0]$ \cite{ferguson2012optimal}.

\section{Simulation Results}

\label{sec:sim}

In this section, we provide simulation results to validate the analytical results. An abstract link-layer model with spectral efficiency given by (\ref{eq:26}) is used in the simulation.
\begin{align}
R = \log(1 + \|h\|^2 \snr_{\textrm{avg}} ) ,
\label{eq:26}
\end{align}
where $\|h\|^2 \sim \textrm{Gamma}(k, 1)$ models fast fading and $ \snr_{\textrm{avg}}$ denotes the average received SNR. The following parameters are used unless otherwise specified: $\snr_{\textrm{avg}}=10$ dB, $k=1$, $\tau=0.1$,  $T_{\textrm{ecca}} = 20$ $\mu s$, $q=32$, $T_{\textrm{cot}, \max}= \frac{12}{32}q$ $m s$, $W=1$ MHz.

In Figure \ref{fig:2}, we study how throughput performance varies with stopping spectral efficiency threshold under different system parameters. For each set of  system parameters $(T_{\textrm{cot}, \max}, q, p)$, Figure \ref{fig:2} shows that there exists an optimal stopping threshold that achieves the maximum throughput. Simulation results in Figure \ref{fig:2} further validate Proposition \ref{pro:2}, i.e., the optimal throughput $\lambda^\star$ is maximized with $q=32$ and $T_{\textrm{cot}, \max} = \frac{12}{32}q=12$ $m s$.\footnote{Here we choose the largest integer for $T_{\textrm{cot}, \max}$ such that $T_{\textrm{cot}, \max} < \frac{13}{32}q$ $m s$, while theoretically $T_{\textrm{cot}, \max}$ can be arbitrarily close to $\frac{13}{32}q$ $m s$.} However,  it can be seen that the gap between two curves with the same $(T_{\textrm{cot}, \max}, p)$ but different $q$'s is quite small. In contrast, the gap between two curves with the same $(q, p)$ but different $T_{\textrm{cot}, \max}$'s is much larger, showing that the performance is more sensitive to $T_{\textrm{cot}, \max}$.

\begin{figure}
\centering
\includegraphics[width=8.5cm]{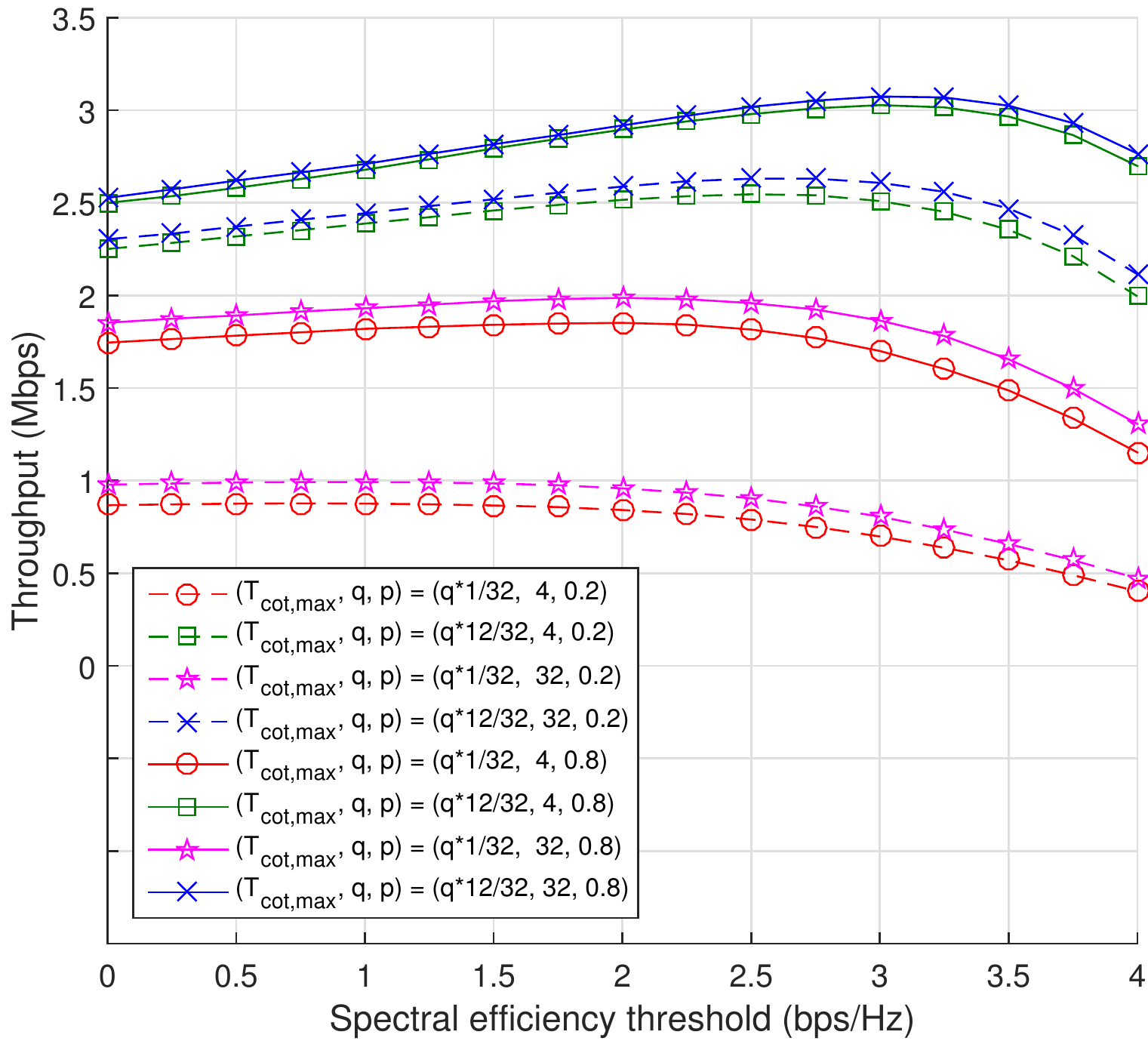}
\caption{Impact of system parameters on throughput performance.}
\label{fig:2}
\end{figure}

In Figure \ref{fig:3}, we compare the throughput performance attained by optimal stopping to the performance of a baseline transmission strategy that the LBE always transmits whenever allowed. Not surprisingly, optimal stopping outperforms the baseline scheme in the entire range of probabilities of clear ECCA check. Note that the value $k$ may be interpreted as the diversity order of the transmission. In particular, with $k=1$ the transmission experiences Rayleigh fading. Figure \ref{fig:3} shows that optimal stopping yields more throughput gains when the operating channel is more versatile (i.e., $k$ is smaller). Figure \ref{fig:3} also shows that optimal stopping yields more throughput gains when $p$ is larger. Interestingly, compared to the baseline scheme, the throughput performance under optimal stopping is much less sensitive with respect to $k$.

%Gain: k=1: 10\% to 23\%
%Gain: k=2: 3\% to 11\%
%Gain: k=4: 1\% to 5\%

\begin{figure}
\centering
\includegraphics[width=8.5cm]{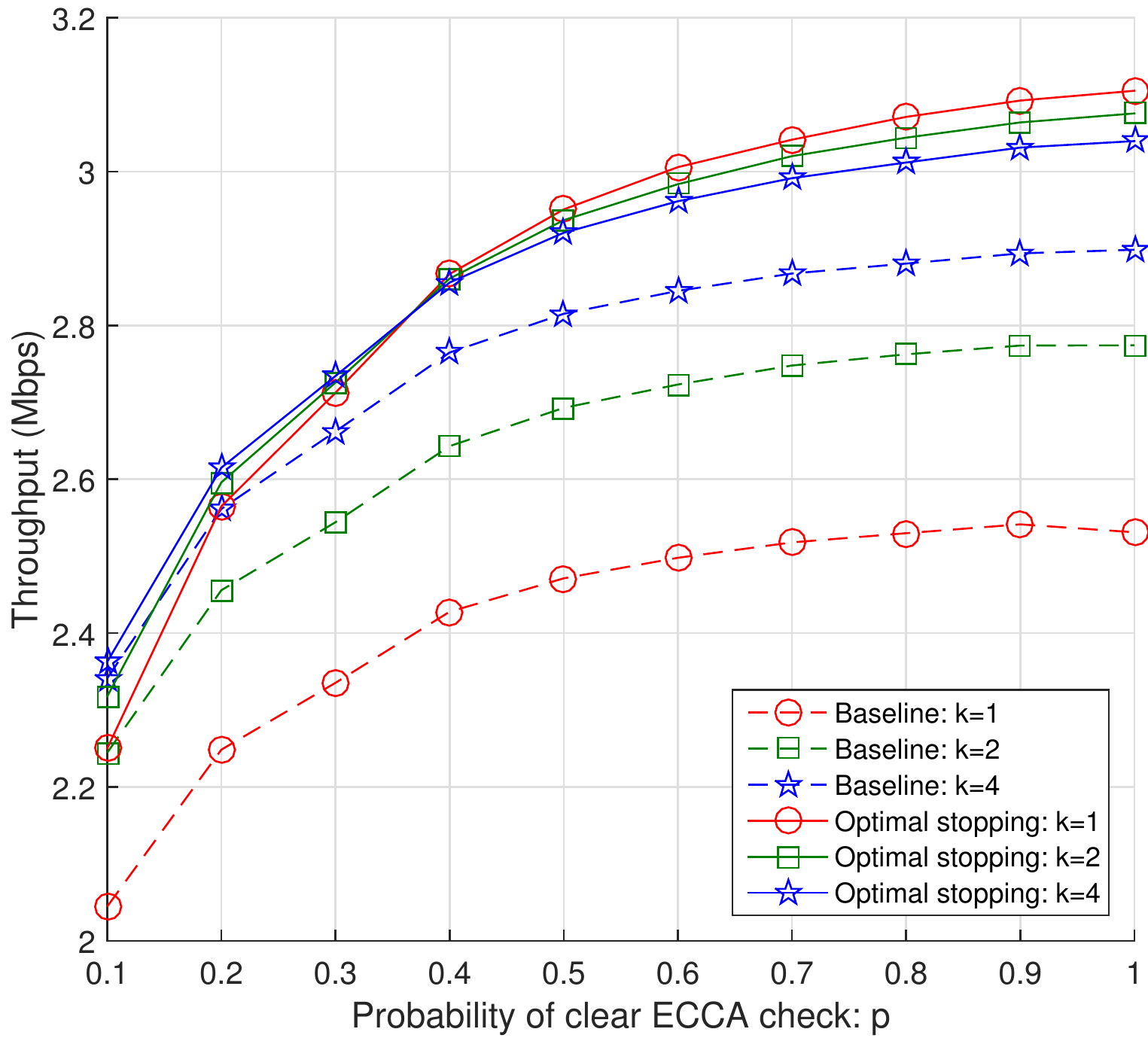}
\caption{Throughput performance: Optimal stopping versus baseline.}
\label{fig:3}
\end{figure}

\section{Conclusions}

\label{sec:conclusions}

In this paper, we have studied throughput optimal LBT transmission for load based equipment in unlicensed spectrum. We show that the throughput optimal strategy is a pure threshold policy and reveal the optimal set of LBT parameters for load based equipment. We have assumed that the outcomes of clear channel check in unlicensed spectrum are i.i.d.. Future work can consider relaxing this assumption to explore more sophisticated radio environment.

%\appendix[Proof of Proposition \ref{pro:1}]

\bibliographystyle{IEEEtran}
\bibliography{IEEEabrv,Reference}

\end{document}